\documentclass[conference]{IEEEtran}
\IEEEoverridecommandlockouts
\usepackage{cite}
\usepackage{amsmath,amssymb,amsfonts}
\usepackage{algorithmic}
\usepackage{graphicx}
\usepackage{textcomp}
\usepackage{mathrsfs}
\usepackage{amsthm}
\theoremstyle{plain}
\newtheorem{theorem}{Theorem}[section]
\newtheorem{lemma}[theorem]{Lemma}

\newtheorem{corollary}[theorem]{Corollary}
\theoremstyle{definition}
\newtheorem{definition}{Definition}
\theoremstyle{remark}
\newtheorem{remark}{Remark}
\usepackage{xcolor}
\def\BibTeX{{\rm B\kern-.05em{\sc i\kern-.025em b}\kern-.08em
    T\kern-.1667em\lower.7ex\hbox{E}\kern-.125emX}}
\usepackage[]{geometry}
\allowdisplaybreaks
\begin{document}
\title{\vspace{6.65mm}Distributed Estimation by Two Agents with Different Feature Spaces}
\author{Aneesh Raghavan and Karl Henrik Johansson
\thanks{*Research supported by the Swedish Research Council (VR), Swedish Foundation for Strategic Research (SSF),  and the Knut and Alice Wallenberg Foundation. The authors are with the Division of Decision and Control Systems, 
Royal Institute of Technology,  KTH,  Stockholm. 
Email: {\tt\small aneesh@kth.se,  kallej@kth.se}}%
}
\maketitle
\begin{abstract}
We consider the problem of estimation of a function by a system consisting of two agents and a fusion center. The two agents collect data comprising of samples of an independent variable and the corresponding value of a dependent variable. The objective of the system is to collaboratively estimate the function without any exchange of data among the members of the system. To this end,  we propose the following framework. The agents are given a set of features using which they construct suitable function spaces to formulate and solve the estimation problems locally. The estimated functions are uploaded to a fusion space where an optimization problem is solved to fuse the estimates (also known as meta-learning) to obtain the system estimate of the mapping. The fused function is then downloaded by the agents to gather knowledge about the other agents estimate of the function. With respect to the framework, we present the following: a systematic construction of fusion space given the features of the agents; the derivation of an uploading operator for the agents to upload their estimated functions to a fusion space; the derivation of a downloading operator for the fused function to be downloaded. Through an example on least squares regression, we illustrate the distributed estimation architecture that has been developed. 
\end{abstract}
\section{Introduction}\label{section 1}
\subsection{Motivation}
In statistical inference problems, observations about a given phenomenon can be obtained using different kind or multiple sensors. The observations are often collated to obtain a single data set. Each sensor output is considered as a \textit{modality} associated with the data set. When the data set is used to estimate a mapping from a set of independent variables to a set of dependent variables, the resulting algorithms are known as \textit{multimodal learning algorithms}.\cite{lahat2015multimodal} provides an overview on to multimodal data fusion focusing on why it is needed and how it can be achieved. \cite{ramachandram2017deep} is a survey paper on deep multimodal learning covering many aspects including comparison with conventional multimodal learning, fusion structures, and applications. Multi-modal learning has found applications in many areas including human activity recognition \cite{ofli2013berkeley} \cite{chen2015utd}, autonomous driving \cite{geiger2013vision} \cite{maddern20171}, and, health monitoring \cite{banos2015design}.      

Kernel methods have played a central role in inference problems. For classical literature on the application of kernel methods to estimation, we refer to \cite{berlinet2011reproducing}, \cite{steinwart2008support}, and,  the references there in. In \cite{pillonetto2014kernel}, \cite{chen2014system}, and, \cite{espinoza2005kernel}, identification of discrete time and continuous time systems using these methods have been reviewed and studied. In more recent times, multimodal kernel learning methods have been studied and algorithms referred to as ``Multiple Kernel Learning" \cite{gonen2011multiple}, have been developed and applied to problems in object recognition \cite{bucak2013multiple}, disease detection \cite{donini2016multimodal}, etc. Kernel methods for deep learning has been studied in \cite{cho2009kernel}. Understanding deep learning through comparison with kernel based learning has been done in \cite{belkin2018understand}. Hence, kernel methods which have been extensively used in classical inference problems, have evolved, and are relevant in contemporary inference problems as well.

Multimodal learning problems are usually studied through data collation, i.e., it is a centralized learning approach. The objective of this paper is to take a step towards achieving multi-modal learning through a distributed approach. One such scheme that has been studied in the literature in the context of IoTs, etc., is vertical federated learning, \cite{kairouz2021advances}, \cite{li2021survey}, when the number of agents is large. However, distributed schemes with fewer number of agents and emphasis on the learning space itself has not received much attention. We note that in our previous work, \cite{raghavan2023distributed}, we considered a distributed regression problem with \textit{noisy} data by two agents and a fusion center with the agents learning in the same space. 

For motivation, consider the following example. There are two agents (robots) in a given region trying to understand their environment. Each agent is restricted to survey a part of the space that they live in. Each agent is aware of the predominant features of the environment in the region that it surveys. By collecting data, given the features, the agent is able to estimate certain aspects of its environment, though that may not be the complete picture. By fusing the knowledge that they have learnt they could arrive at a complete knowledge of their environment. Given the fused knowledge, they should be able to interpret that in their own knowledge space, i.e., they should be able to express it using their own feature vectors to reflect upon the other agents perception of the environment. It could enhance their perception of the environment. 
\subsection{Problem Considered}
Though we do not formally define data, information, knowledge in the context of inference problems, we differentiate between them as following. A set of measurable outcomes associated with an observable phenomenon or an experiment is referred to as \textit{data}. Structured data, that is, data which could  used to infer models or hidden patterns is referred to as \textit{information}. The inferred models or patterns are referred to as knowledge. In the context of the experiment, the set of all possible models or patterns is referred to as the \textit{knowledge space} (KS). As information is received sequentially, the knowledge about the observed phenomenon evolves in the KS. In the simplest setting, we can consider the estimation of a mapping from an independent variable (input) to a dependent variable (output). Information would correspond to pairs of input, output measurements. Knowledge corresponds to the function from the input to the output, while the function space where the inference problem is studied is interpreted as a knowledge space. Other examples of knowledge spaces include set of probability measures on a Borel $\sigma$ algebra, probability measures on a orthoposet \cite{baras2021order}. 

The problem considered is as follows. We consider two agents and a fusion center. Each agent receives information comprising of samples of an independent variable and corresponding values of the dependent variable. The underlying phenomenon generating information for both the agents is the same. Given the information, the objective of the agents and the fusion center is to collaboratively estimate the mapping from the independent variable to the dependent variable \textit{without} exchanging any information between the agents or the agents and the fusion center. Along with the samples, each agent is provided a set of features predominant in the information received by them. 
\vspace{-0.4cm}
\subsection{Contributions}
We propose the following architecture, refer Figure \ref{Figure 1}: Given local information, each agent estimates the mapping from the input to the output  by solving a least squares regression problem in its local KS. The functions estimated by both agents are uploaded to the fusion center. At the fusion center, a fusion problem as an optimization problem is formulated and solved to fuse the functions estimated by the agents. The fused function is considered as the function estimated by the system. The fused function is downloaded on to the local KSs by the agents and is considered as the final estimates of the agents. 
\begin{figure}
\begin{center}
\includegraphics[width=\columnwidth]{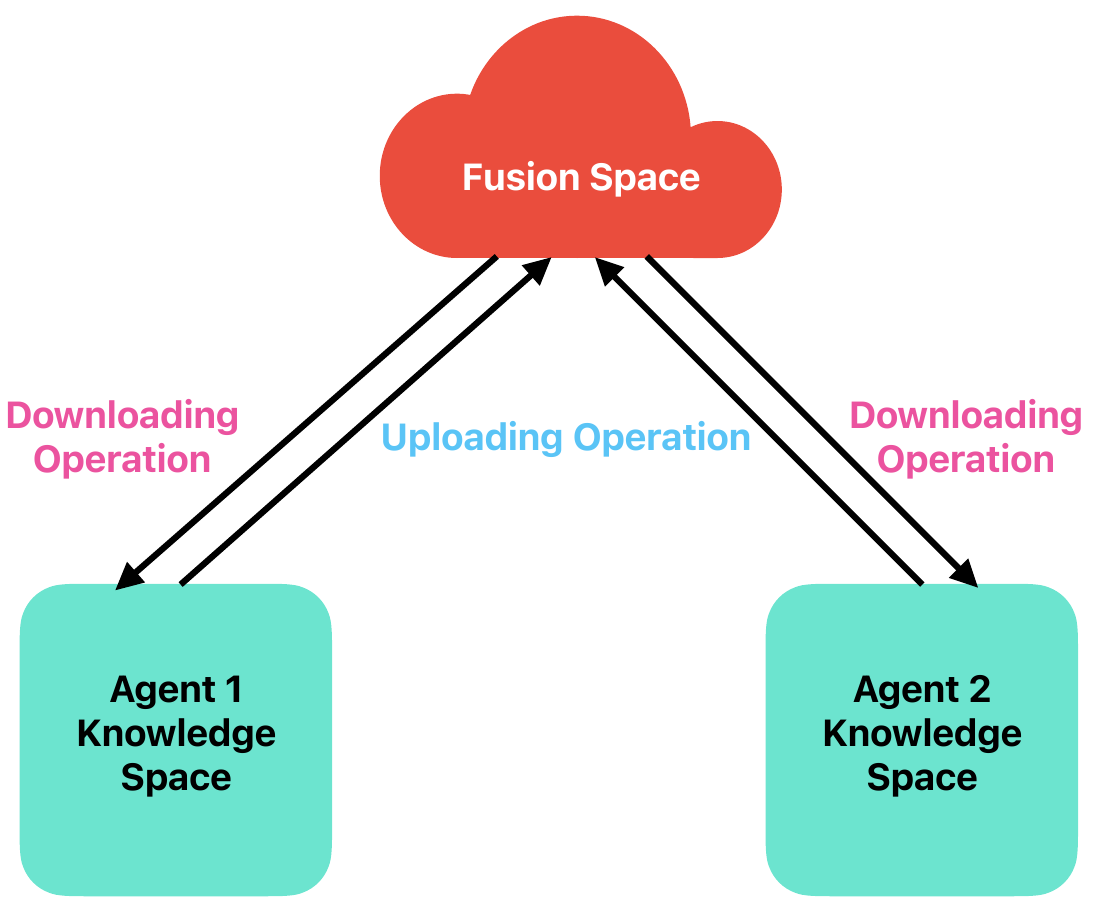}
\caption{Schematic for Distributed Estimation Architecture} 
\label{Figure 1}
\end{center}
\vspace{-1cm}
\end{figure}  

With respect to the above frame work we prove the following: Given the features, we present the construction of the individual KSs of the agents as a reproducing kernel Hilbert space (RKHS). The fusion space is defined as the set of functions obtained through linear combination of functions in the the local KSs. We prove that the fusion space is also an RKHS whose kernel is the sum of the kernels of the agents. As a corollary, we obtain that the uploading operator used by the agents to upload functions from the local KS to the fusion space is a linear bounded operator. All functions in the fusion center might not be decipherable in the local KSs, a download operator is needed to suitably transform the fused function which can be interpreted in the local KS.  We present a detailed construction of such an download operator and prove that it is linear and bounded. To illustrate the  distributed estimation scheme, we present a numerical example. 
\subsection{Novelty and Interpretation of Proposed Solution}
The novelty of the proposed solution is that distributed learning can be achieved with different agents learning in different function spaces. The framework is applicable to other learning problems including classification and density estimation. The function spaces can be chosen based on the predominant features in the data collected by them. We focus on the construction of the learning spaces and how functions can be transferred from one space to another and less on the estimation procedure itself, i.e., the focus of the paper is on framework and the structure for distributed learning, less on the learning algorithm itself. To emphasize the same, the example presented here demonstrates how a least squares estimation problem over polynomials and exponentials can be solved in distributed way; one agent focusing on polynomials and other focusing on exponentials. 

In kernel methods, the kernels used for learning are usually assumed to be known. Depending on the data collected by an agent, there could be a predominant feature in the data. This could dictate the kernel used by the agent for estimation. By using different kernels the agents are tuning their estimation procedure to the specific data collected by them. Further by using different kernels they could be exhausting the set of kernels that could be used for the given data set, if one is not able to make a clear choice upfront. Thus, from a practical standpoint the system could benefit when the agents chose different kernels. 

There is an additional ``benefit" of privacy as, the proposed framework promotes local processing of data, thus keeping the data private to the agent. An interpretation for distributed learning is ``robust" learning. Though this to be formally proven, different agents estimating the function using different kernels ensures that the estimation scheme is robust to noise in data (as same noise need not be observed by the agents) and to the kernels chosen. If all the agents where to chose the same kernel, there could be a potential bias in the estimation procedure. By choosing different kernels, it is ensured that  different features of the data is given importance by different agents.

The paper is organized as follows. It is different from the order in which the contributions were discussed. In Section \ref{section 2}, we discuss the construction of the individual KSs, the fusion space, the derivation of the uploading and downloading operator. In Section \ref{section 3}, we discuss the regression problem for the agents, its solution, and, the fusion problem in fusion space. In Section \ref{section 4}, we present a numerical example demonstrating the distributed estimation methodology. In Section \ref{section 5}, we summarize the key aspects of this paper and focus on future work. Notation: we use superscript for the agent, subscript for samples and summation indicies. We represent vectors obtained by concatenating smaller vectors in boldface. For a function $f\in V$, $V$ vector space, we use the notation $f$ when it is treated as a vector and the notation $f(\cdot)$ when it is treated as a function. The projection onto a subspace $\mathcal{M}$ of a Hilbert space $H$ is denoted by  $\Pi_{\mathcal{M}}$. 
\section{Construction of Knowledge Spaces}\label{section 2}
In this section, we discuss the construction of KSs for the individual agents and the fusion space. Let $\mathcal{X} \subset \mathbb{R}^d$. Let $(H, \langle \cdot, \cdot \rangle_{H})$ be a Hilbert space of functions, $f: \mathcal{X} \to \mathbb{R}$. Let $K: \mathcal{X} \times \mathcal{X} \to \mathbb{R}$ be a function, $n \in \mathbb{N}$, $\{x_1, \ldots, x_n \} \subset \mathcal{X}$, and, $\mathbf{K}:=(K(x_{i}, x_{j}))_{ij}$ be Gram (kernel) matrix of $K$ with respect to $x_1, \ldots, x_n$.
\begin{definition}\label{defitnion 1}
The function $K(\cdot,\cdot)$ is said to be a positive definite kernel if the gram matrix generated by the function is a positive definite matrix for all $n$ and $\{x_1, \ldots, x_n \} \subset \mathcal{X}$. 
\end{definition}
\begin{definition}\label{definition 2}
$(H, \langle \cdot, \cdot \rangle_{H})$ is said to be a reproducing kernel Hilbert space (RKHS) with a positive definite kernel $K$, if, 
\begin{itemize}
\item $K(\cdot,x) \in H, \; \forall x\in \mathcal{X},$
\item the reproducing property is satisfied
\begin{align*}
f(y) = \langle f(\cdot), K(\cdot,y) \rangle_{H}, f\in H, y \in \mathcal{X}.
\end{align*}
\end{itemize}
\end{definition}
\subsection{Construction of Individual Knowledge Spaces for the Agents}
The set of features for agent $i$ is the set of functions, $\{\varphi^{i}_{j}(\cdot)\}_{j \in \mathcal{I}^i}$, where $\varphi^{i}_{j} : \mathcal{X} \to \mathbb{R}$ and $| \mathcal{I}^i | < \infty$. The knowledge space for agent $i$ is the finite dimensional vector space, $H^i$, defined as:
\begin{align*}
H^{i} = \{f : f(\cdot) = \sum_{j \in \mathcal{I}^i}\alpha_j\varphi^i_j(\cdot), \{\alpha_j\}_{j \in \mathcal{I}^i} \subset \mathbb{R}\}.
\end{align*}
The null vector for the space $H^i$ is the function $\theta^i(\cdot)$ defined as $\theta^i(x) = 0, \forall x \in \mathcal{X}$. For agent $i$, the features are assumed to be linearly independent, i.e., $\sum_{j \in  \mathcal{I}^i} \alpha_{j} \varphi^{i}_{j}= \theta^i$ if and only if $\alpha_j =0,\;  \forall j$. The function space, $H^i$, is equipped with the inner product, $\langle \cdot, \cdot \rangle_{H^i}: H^{i} \times H^{i} \to \mathbb{R}$, defined as follows. For, $f(\cdot) =  \sum_{j \in \mathcal{I}^i}\alpha_j\varphi^i_j(\cdot), g(\cdot)  = \sum_{j \in \mathcal{I}^i}\beta_j\varphi^i_j(\cdot)$,
\begin{align*}
\langle f(\cdot), g(\cdot) \rangle_{H^i}  \hspace{-2pt} = \hspace{-2pt} \langle \sum_{j \in \mathcal{I}^i}\alpha_j\varphi^i_j(\cdot) , \sum_{j \in \mathcal{I}^i}\beta_j\varphi^i_j(\cdot)  \rangle_{H^i} := \sum_{j \in \mathcal{I}^i}\alpha_j\beta_j.
\end{align*} 
It can be verified that the above definition of inner product on $H^i$ satisfies the axioms of a inner product on a vector space. The norm induced by the inner product is $|| f ||^2_{H^i} = \sum_{j \in \mathcal{I}^i}\alpha_j^2$. The kernel $K^i: \mathcal{X} \times \mathcal{X} \to \mathbb{R}$ is defined as, 
\begin{align*}
K^{i}(x,y) = \sum_{j \in \mathcal{I}^i}\varphi^{i}_{j}(x)\varphi^{i}_{j}(y).
\end{align*}
Let $\{x_{1}, \ldots, x_{n}\}\subset \mathcal{X}$. Let $\mathbf{K^i}:=(K^i(x_{k}, x_{l}))_{kl}$, be the Gram matrix of $K^i(\cdot, \cdot)$ with respect to $\{x_1, \ldots, x_n\}$. Then for any $\boldsymbol{\alpha} \in \mathbb{R}^n$, $\boldsymbol{\alpha}^{T}\mathbf{K^i}\boldsymbol{\alpha}=$
\begin{align*}
&\sum^{n}_{k=1}\sum^{n}_{l=1} \alpha_k\alpha_l K^i(x_k,x_l) = \sum^{n}_{k=1}\sum^{n}_{l=1}\alpha_k \alpha_l\sum_{j \in \mathcal{I}^i}\varphi^{i}_{j}(x_k)\varphi^{i}_{j}(x_l)\\
&= \sum_{j \in \mathcal{I}^i} \Big( \sum^{n}_{k =1}\alpha_k \varphi^{i}_{j}(x_k) \Big) \Big( \sum^{n}_{l =1}\alpha_l \varphi^{i}_{j}(x_l) \Big) = || f ||^2_{H^i} \geq 0,
\end{align*}
where $f(\cdot) = \sum_{j \in \mathcal{I}^i} \Big(\sum^{n}_{k =1}\alpha_k \varphi^{i}_{j}(x_k)\Big) \varphi^{i}_{j}(\cdot)$. The Gram matrix of $K^{i}(\cdot, \cdot)$ is positive definite for any  $\{x_{1}, \ldots, x_{n}\}\subset \mathcal{X}$, for all $n \in \mathbb{N}$. Thus, $K^i(\cdot,\cdot)$ is a positive definite kernel (From Definition 1, also refer \cite{hofmann2008kernel}). We note that, $K^{i}(\cdot, y) = \sum_{j \in \mathcal{I}^i}\varphi^{i}_{j}(y)\varphi^{i}_{j}(\cdot) \in H^{i}$, with $\alpha^{i}_{j} =\varphi^{i}_{j}(y),\; \forall j \in\mathcal{I}^i$. For $f \in H^i$, $f(\cdot) =   \sum_{j \in \mathcal{I}^i}\alpha_j\varphi^i_j(\cdot)$,
\begin{align*}
\langle f(\cdot), K^i(\cdot,y) \rangle_{H^{i}} &= \langle \sum_{j \in \mathcal{I}^i}\alpha_j\varphi^i_j(\cdot), \sum_{j \in \mathcal{I}^i}\varphi^{i}_{j}(y)\varphi^{i}_{j}(\cdot)  \rangle_{H^{i}}\\
&= \sum_{j \in \mathcal{I}^i}\alpha_j \varphi^{i}_{j}(y) = f(y).
\end{align*}
The reproducing property is satisfied by $f\in H^i$ with kernel $K^{i}(\cdot, \cdot)$. From Definition 2, it follows that $H^{i}$ is a RKHS with kernel $K^{i}(\cdot, \cdot)$.
\subsection{Construction of the Fusion Space}
The motivation for the construction of a fusion space is to build a function space where algebraic operations can be simultaneously performed on functions living in the local KSs of the agents. Hence, at the very least it should include functions from the local KSs. As a possible candidate, one could consider $H^{1} \cup H^{2}$. The properties of the original KSs could be lost as the resulting space need not be vector space or a RKHS. The space $\big(H^{1} \oplus H^{2}\Big) / \sim$ seems to be another reasonable candidate, however a set of equivalent classes of functions cannot be a RKHS. From the assumption, the vectors in $\{\varphi^{1}_{j}(\cdot)\}_{j \in \mathcal{I}^1}$ are linearly independent, and so are the vectors in $\{\varphi^{2}_{j}(\cdot)\}_{j \in \mathcal{I}^2}$. However, it is not necessary that the vectors in $\{\varphi^{1}_{j}(\cdot)\}_{j \in \mathcal{I}^1} \cup \{\varphi^{2}_{j}(\cdot)\}_{j \in \mathcal{I}^2} $ are linearly independent. Since $K^{1}$ and $K^{2}$ are positive definite kernels, $K = K^{1} + K^{2}$ is also a positive definite kernel. This suggests that the fusion space could be constructed as the RKHS corresponding to the kernel $K(\cdot,\cdot)$ by considering the linear combination of all functions in $H^{1}$ and $H^{2}$.
\begin{definition}
The fusion space, $H$, is defined as $H = \{f: f= f^{1} + f^{2}, f^{1} \in H^{1}, f^{2} \in H^{2}\}$.
\end{definition}
In the following theorem, we characterize the fusion space as a RKHS after associating a suitable inner product and find an expression for the norm induced by the inner product. 
\begin{theorem}\label{Theorem 1}
If $K^{i}(\cdot,\cdot)$ is the reproducing kernel of Hilbert space $H^{i}$,  with norm $||\cdot||_{H^i}$, then $K(x,y)=K^{1}(x,y) + K^{2}(x,y)$ is the reproducing kernel of the space $H = \{f: f= f^1 + f^2 | f^{i} \in H^{i}\}$ with the norm:
\begin{align*}
||f||^2_{H} = \underset{\substack{f^1 + f^2 = f,\\ f^{i} \in H^{i}} }  \min \;\;  ||f^{1}||^2_{H^1} + ||f^{2}||^2_{H^2}.
\end{align*} 
\end{theorem}  
\begin{proof}
Let $H_{\prod} = H^{1} \times H^2$ denote the product space, with inner product $\langle (f^{1},f^{2}),(g^{1},g^{2})\rangle_{H_{\prod}} = \langle f^{1}, g^{1}\rangle_{H^1} + \langle f^{2}, g^{2}\rangle_{H^2}$. Let $H =\{f : f= f^{1} + f^{2}, f^{i} \in H^{i}, i=1,2\}$. Clearly, $H$ is a vector space whose null vector we denote by $\theta$. Let $L: H_{\Pi}  \to H$ be a operator defined as $L((f^{1},f^{2})) =f^{1} +f^{2}$. $L$ is a linear operator and its null space, $\mathcal{N}(L) =\{(f^{1}, f^{2}) \in H_{\prod} : f^{1} +f^{2} =\theta \}$ is a closed subspace as it is finite dimensional. The basis vectors for $H_{\prod}$ are given by $\{\varphi^{1}_{j} \times \theta^2 \}_{j \in \mathcal{I}^{1}} \cup \{\theta^1 \times \varphi^{2}_{j} \}_{j \in \mathcal{I}^{2}}$. Any vector in $H_{\prod}$, can be expressed as $(f^1,f^2) = (\sum_{j \in \mathcal{I}^1}\alpha^1_j \varphi^{1}_{j}, \sum_{j \in \mathcal{I}^2}\alpha^2_j \varphi^{2}_{j})$ and thus $H_{\prod}$ is isomorphic to $\mathbb{R}^{\mathcal{I}^1 + \mathcal{I}^2}$. The null space $\mathcal{N}(L)$ is isomorphic to the subspace,
\begin{align*}
\mathcal{N}=\Big\{\Big(\boldsymbol{\alpha^1}, \boldsymbol{\alpha^2} &\Big) \in \mathbb{R}^{\mathcal{I}^1 + \mathcal{I}^2}: \sum_{j \in \mathcal{I}^1}\alpha^1_j \varphi^{1}_{j} + \sum_{j \in \mathcal{I}^2}\alpha^2_j \varphi^{2}_{j}= \theta,\\
&\boldsymbol{\alpha^1} = \Big(\alpha^{1}_{1}, \ldots \alpha^{1}_{\mathcal{I}^1}\Big), \boldsymbol{\alpha^2} = \Big(\alpha^{2}_{1}, \ldots \alpha^{2}_{\mathcal{I}^2}\Big) \Big\}.
\end{align*}
Since $\mathcal{N}(L)$ is a closed subspace, there exists a unique closed subspace $\mathcal{M}$ such that $H_{\prod} = \mathcal{M}\oplus \mathcal{N}(L)$. The mapping $L_{\mathcal{M}} = L \circ \Pi_{\mathcal{M}}$ (operator $L$ restricted to subspace $\mathcal{M}$) is bijection from $\mathcal{M}$ to $H$. For any function $f \in H$, let $L^{-1}_{\mathcal{M}}(f) = (L^{1}(f), L^{2}(f))$, i.e., $(L^{1}(f), L^{2}(f))$ is the \textit{unique} tuple of functions in $\mathcal{M}$ such that $L^{1}(f) \in H^{1}, \; L^{2}(f)\in H^{2}$ and $L^{1}(f) + L^{2}(f) =f$. We now define the inner product on $H$ as follows:
\begin{align*}
\langle f, g \rangle_{H} = \langle L^{1}(f), L^{1}(g) \rangle_{H^{1}} +  \langle L^{2}(f), L^{2}(g) \rangle_{H^{2}}.
\end{align*}
Since $K^{i}(\cdot, y) \in H^i$, it follows that $K(\cdot, y) = K^{1}(\cdot, y) + K^2(\cdot, y) \in H$. We claim that $L^{i}(K(\cdot, y))= K^{i}(\cdot, y)$. Indeed, since $K(\cdot, y) = \sum_{j \in \mathcal{I}^1} \varphi^{1}_{j}(y)\varphi^{1}_{j}(\cdot) + \sum_{j \in \mathcal{I}^2}\varphi^{2}_{j}(y) \varphi^{2}_{j}(\cdot)$, to prove the claim it suffices to prove that the vector $\Big( \varphi^{1}_{1}(y), \ldots, \varphi^{1}_{ \mathcal{I}^1}(y), \varphi^{2}_{1}(y), \ldots, \varphi^{2}_{ \mathcal{I}^2}(y) \Big)$ is orthogonal to $\mathcal{N}$ for all $y \in \mathcal{X}$. That is, $ \forall \; \Big(\boldsymbol{\alpha^1}, \boldsymbol{\alpha^2} \Big) \in \mathcal{N}$, 
\begin{align*}
\sum_{j \in \mathcal{I}^1}\alpha^1_j \varphi^{1}_{j}(y) + \sum_{j \in \mathcal{I}^2}\alpha^2_j \varphi^{2}_{j}(y)= 0,
\end{align*} 
which is true from the definition of $\mathcal{N}$. Thus, 
\begin{align*}
\langle f(\cdot), K(\cdot,y) \rangle_{H} &= \langle L^{1}(f)(\cdot), L^{1}( K(\cdot,y) ) \rangle_{H^{1}} +  \langle L^{2}(f)(\cdot), \\
L^{2}( K(\cdot,y) ) \rangle_{H^{2}} &= \langle L^{1}(f)(\cdot), K^{1}(\cdot,y) ) \rangle_{H^{1}} +  \langle L^{2}(f)(\cdot), \\
K^{2}(\cdot,y) ) \rangle_{H^{2}} &= L^{1}(f)(y) + L^{2}(f)(y) = f(y),
\end{align*}
satisfying the reproducing property. From Definition \ref{definition 2}, it follows that $(H, \langle \cdot, \cdot \rangle_{H})$ is a RKHS with kernel $K(\cdot,\cdot)$. For $f \in H$, let $(f^{1}, f^{2}) \in H_{\prod}$ be such that $f = f^{1} + f^{2} = L^{1}(f) + L^{2}(f)$. We note that $(L^{1}(f), L^{2}(f)) = \Pi_{\mathcal{M}}((f^{1}, f^{2}))$.  Computing norm of $f$,
\begin{align*}
&|| f||^{2}_{H} = \langle f, f \rangle_{H} = || L^{1}(f) ||^{2}_{H^1} + || L^{2}(f) ||^{2}_{H^2}\\
&|| f^{1}||^{2}_{H^1} + ||f^{2}||^{2}_{H^2} = || (f^{1},f^{2}) ||^{2}_{H_{\prod}} = \underbrace{||  \Pi_{\mathcal{M}}((f^{1}, f^{2}))} + \\
&  \Pi_{\mathcal{N}(L)}((f^{1}, f^{2}))||^{2}_{H_{\prod}}. \hspace{1.5cm}= || (L^{1}(f),L^{2}(f))||^{2}_{H_{\prod}}
\end{align*}
Thus, for $f^{1}, f^{2}$ such that $f= f^{1}+f^{2}$, the minimum of $|| f^{1}||^{2}_{H^1} +  ||f^{2}||^{2}_{H^2}$ is achieved when $ \Pi_{\mathcal{N}(L)}((f^{1}, f^{2}))||^{2}_{H_{\prod}} = \theta$, i.e., $f^{i} = L^{i}(f)$, and, is equal to $|| f||^{2}_{H}$.
\end{proof}
Given $f \in H^{1}$, we let $f^{1} = f $ and $f^{2} = 0$. From the Theorem \ref{Theorem 1}, we conclude that $||f||_{H} \leq || f ||_{H^1}$. Similarly $|| f ||_{H} \leq || f ||_{H^2}, f \in H^2$.
\begin{corollary}
The uploading operator from agent $i$'s knowledge space, $H^i$,  to the fusion space $H$, $\hat{L}^{i}: H^{i} \to H$, is $\hat{L}(f) = f$. $\hat{L}^{i}(\cdot)$, is linear and is bounded, $|| \hat{L}^{i} || = \sup\{ || f||_{H} : f\in H^{i}, || f ||_{H^{i}} =1 \} \leq 1$.
\end{corollary}
\subsection{Retrieval of Individual Knowledge Spaces from Fusion Space}\label{subsection 2.3}
Given a function in the fusion space, it is not necessary that it belongs to both the KSs. Hence, it is mandatory to transform it to a form where it can be expressed using the features of the local KS. The objective of this subsection is to find an operation which would transform the function onto the individual KSs. Given a function in the fusion space, one possibility is to express the function as linear combination of $\{\varphi^{1}_{j}(\cdot)\}_{j \in \mathcal{I}^1} \cup \{\varphi^{2}_{j}(\cdot)\}_{j \in \mathcal{I}^2}$ and then simply set the coefficients of the functions which do not belong to a local KS to zero. The resulting function need not belong to the local KS, i.e., need not satisfy the reproducing property of the KS. To ensure the same, in the following derivation every function in either of the local KS can be retrieved from a function in the fusion space. The operation corresponding to one of the KS when applied to a function in the fusion space would only result in a function in the KS.
\begin{lemma}\label{Lemma 1}
For every $f\in H$, there exits $\{y_{k,f}\}^{n}_{k=1}$ and $\{\beta_{k,f}\}^{n}_{k=1}$ such that,
\begin{align*}
f(\cdot) \hspace{-3pt}= \hspace{-2pt} \sum_{i=1,2}\sum_{j \in \mathcal{I}^{i}}\sum^{n}_{k=1}\beta_{k,f}\varphi^{i}_{j}(y_{k,f})\varphi^{i}_{j}(\cdot) =  \sum^{n}_{k=1}\beta_{k,f}K(\cdot,y_{k,f}).
\end{align*}
\end{lemma}
\begin{proof}
Given $f\in H$, we characterize $L^{1}(f)$ and $L^{2}(f)$. Let, 
\begin{align*}
\Phi(y) \hspace{-2pt} = \hspace{-2pt}[\varphi^{1}_{1}(y), \ldots, \varphi^{1}_{\mathcal{I}^1}(y),\varphi^{2}_{1}(y), \ldots, \varphi^{2}_{\mathcal{I}^1}(y)] \hspace{-2pt} \in \hspace{-2pt}  \mathbb{R}^{\mathcal{I}^1 + \mathcal{I}^2}\hspace{-3pt}, \;
\end{align*}
$y \in \mathcal{X} $. Let $\hat{\mathcal{M}}$ be the span of $\{\Phi(y)\}_{y\in \mathcal{X}}$, i.e., $\hat{\mathcal{M}} = \{\boldsymbol{\gamma} \in \mathbb{R}^{\mathcal{I}^1 + \mathcal{I}^2}\hspace{-3pt}:  \boldsymbol{\gamma} = \sum^{n}_{k=1} \beta_{k}\Phi(y_k), n \in \mathbb{N}, \{\beta_{k}\}^n_{k =1} \subset \mathbb{R}\}$, and $\bar{\mathcal{M}}$ be such that $ \mathbb{R}^{\mathcal{I}^1 + \mathcal{I}^2} = \bar{\mathcal{M}} \oplus \mathcal{N}$. We claim that $\bar{\mathcal{M}} = \hat{\mathcal{M}}$. Indeed, from the definition of $\mathcal{N}$, it follows that for every $\boldsymbol{\alpha} \in \mathcal{N}$, $\boldsymbol{\alpha}^{T}\Phi(y) = 0, \; \forall y \in \mathcal{X}$. Thus, if $ \boldsymbol{\gamma} \in \hat{\mathcal{M}}$, then $\boldsymbol{\gamma}^{T}\boldsymbol{\alpha} = 0, \; \forall \boldsymbol{\alpha} \in \mathcal{N}$, i.e., $\boldsymbol{\gamma} \in \bar{\mathcal{M}}$. Suppose, $\boldsymbol{\bar{\gamma}} \in \bar{\mathcal{M}}$ is equal to $\boldsymbol{\gamma} + \boldsymbol{v}$, where $\boldsymbol{\gamma} \in \hat{\mathcal{M}}$ and $\boldsymbol{v} \in \mathbb{R}^{\mathcal{I}^1 + \mathcal{I}^2}, \boldsymbol{v} \neq \theta, \boldsymbol{v} \notin \hat{\mathcal{M}}$. Then, $\boldsymbol{\bar{\gamma}}^T \boldsymbol{\alpha} = 0, \; \boldsymbol{\alpha} \in \mathcal{N}$ which implies $\boldsymbol{v}^T\boldsymbol{\alpha}=0, \; \boldsymbol{\alpha} \in \mathcal{N}$. This is a contradiction as $\boldsymbol{\alpha} \in \mathcal{N}$ if and only if $\boldsymbol{\alpha}^T\Phi(y) = 0, \; \forall y \in \mathcal{X}$. The condition $\boldsymbol{v}^T\boldsymbol{\alpha}=0, \; \boldsymbol{\alpha} \in \mathcal{N}$ introduces an additional constraint which is not satisfied, except when $\boldsymbol{v} = \theta$. Thus, $\boldsymbol{\bar{\gamma}} = \boldsymbol{\gamma} \in \hat{\mathcal{M}}$. Since $\hat{\mathcal{M}} = \bar{\mathcal{M}}$ is isomorphic to $\mathcal{M}$, every vector in $\mathcal{M}$ is given by 
\begin{align*}
&\Big(\hspace{-2pt}\sum_{j\in \mathcal{I}^{1}}\gamma^{1}_{j}\big(\varphi^{1}_{j} \times \theta^2\big) + \sum_{j\in \mathcal{I}^{2}}\gamma^{2}_{j}\big(\theta^1  \times \varphi^{2}_{j}\big)\Big) = \Big(\hspace{-3pt}\sum_{j\in \mathcal{I}^{1}}\gamma^{1}_{j}\varphi^{1}_{j},\\ 
&\sum_{j\in \mathcal{I}^{2}}\gamma^{2}_{j}\varphi^{2}_{j}\Big) \hspace{-2pt}= \hspace{-3pt}\Big( \hspace{-4pt}\sum_{j\in \mathcal{I}^{1}}\sum^{n}_{k=1}\beta_k\varphi^{1}_{j}(y_k)\varphi^{1}_{j}, \sum_{j\in \mathcal{I}^{2}} \sum^{n}_{k=1}\beta_k\varphi^{2}_{j}(y_k)\varphi^{2}_{j}\Big), 
\end{align*}
where $\boldsymbol{\gamma} = [\gamma^{1}_{1}, \ldots, \gamma^{1}_{\mathcal{I}^1}, \gamma^{2}_{1}, \ldots, \gamma^{2}_{\mathcal{I}^2}] \in \bar{\mathcal{M}}$ and $\gamma^{i}_{j} =\sum^{n}_{k=1}\beta_{k}\varphi^{i}_{j}(y_k)$. Thus, for any $f \in H$, $(L^{1}(f), L^{2}(f)) = \Big(\sum^{n}_{k=1}\beta_{k,f}K^{1}(\cdot, y_{k,f}), \sum^{n}_{k=1}\beta_{k,f}K^{2}(\cdot, y_{k,f})\Big)$. This implies that, $f = L^{1}(f) + L^{2}(f) =  \sum^{n}_{k=1}\beta_{k,f}K(\cdot,y_{k,f})$.
\end{proof}
\begin{lemma}\label{Lemma 2}
Given the RKHS, $(H, \langle \cdot,\cdot \rangle_{H})$, with kernel $K(\cdot,\cdot)$ and the kernels $K^{i}(\cdot,\cdot),\;i=1,2$, such that $K(x,y) = K^{1}(x,y) + K^{2}(x,y)$, we define operators, $\bar{L}^{i}: H \to H$, as
\begin{align*}
\bar{L}^{i}(f)(x) =\langle f(\cdot), K^{i}(\cdot,x) \rangle_{H}, \text{ for}, i=1,2.
\end{align*}
Then, $\bar{L}^{i}$ is linear, symmetric, positive and bounded, $|| \bar{L}^{i} ||\leq 1$.
\end{lemma}
\begin{proof}
From the linearity of the inner product with respect to the first argument, it follows that $\bar{L}^{i}$ is linear. From equation set (\ref{Equation 1}), we note that $\bar{L}^{i}$ is positive, i.e., $\langle \bar{L}^{i}(f), f \rangle_{H} \geq 0, \; \forall f \in H$.  From equation set(\ref{Equation 2}), we note that $\bar{L}^{i}$ is symmetric, i.e., $\langle \bar{L}^{i}(f), g \rangle_{H} = \langle f, \bar{L}^{i}(g) \rangle_{H}, \; \forall f,g \in H$. Since $\bar{L}^{1}$ and $\bar{L}^{2}$ are symmetric, positive, and their sum is the identity operator ($\bar{L}^{1} + \bar{L}^{2} = \mathbb{I}$),  $0\leq || \bar{L}^{i} || \leq 1$.
\begin{figure*}
\hrulefill 
\begin{align}
f = & \sum^{n}_{j=1}\beta_{j}K(\cdot, y_{j}), \bar{L}^{i}(f)(x)  =\langle \sum^{n}_{j=1}\beta_{j}K(\cdot, y_{j}) , K^{i}(\cdot, x) \rangle_{H}  = \sum^{n}_{j=1}\beta_{j}\langle K^{i}(\cdot, x) ,  K(\cdot, y_{j}) \rangle_{H}= \sum^{n}_{j=1}\beta_{j}K^{i}(y_{j}, x)\nonumber\\
=&\sum^{n}_{j=1}\beta_{j}K^{i}(x, y_{j}). \; \langle \bar{L}^{i}(f), f \rangle_{H} \hspace{-2pt} =  \hspace{-2pt}\langle \sum^{n}_{j=1}\beta_{j}K^{i}(\cdot, y_{j}) , \sum^{n}_{k=1}\beta_{k}K(\cdot, y_{k}) \rangle_{H} =\sum^{n}_{j=1}\sum^{n}_{k=1}\beta_{j}\beta_{k}K^{i}(y_k,y_j) = ||\bar{L}^i(f)||^2_{H^{i}} \geq 0.\label{Equation 1} \\ 
g \hspace{-2pt}= \hspace{-2pt} &\sum^{m}_{l=1}\delta_{l}K(\cdot, y_{l}), \langle \bar{L}^{i}(f), g \rangle_{H}\hspace{-2pt} = \hspace{-2pt} \langle \sum^{n}_{j=1}\beta_{j}K^{i}(\cdot, y_{j}), \sum^{m}_{l=1}\delta_{l}K(\cdot, y_{l}) \rangle_{H} \hspace{-2pt} = \hspace{-2pt}\sum^{n}_{j=1}\sum^{m}_{l=1}\beta_{j}\delta_{l}K^{i}(y_{l},y_{j}).\langle f, \bar{L}^{i}(g) \rangle_{H} \hspace{-2pt} = \hspace{-4pt}  \langle  \sum^{n}_{j=1}\beta_{j}K(\cdot, y_{j}), \nonumber\\
\sum^{m}_{l=1}&\delta_{l}K^{i}(\cdot, y_{l})\rangle_{H} \hspace{-3pt}  = \hspace{-4pt}  \sum^{n}_{j=1}\sum^{m}_{l=1}\beta_{j}\delta_{l}\langle K^{i}(\cdot, y_{l}), K(\cdot, y_{j})\rangle_{H}\hspace{-2pt}  = \hspace{-4pt} \sum^{n}_{j=1}\sum^{m}_{l=1}\beta_{j}\delta_{l}K^{i}(y_j, y_{l}) \hspace{-2pt}  = \hspace{-2pt}  \sum^{n}_{j=1}\sum^{m}_{l=1}\beta_{j}\delta_{l}K^{i}(y_l, y_{j}) \hspace{-3pt}  =  \hspace{-3pt}  \langle \bar{L}^{i}(f), g \rangle_{H}.  \label{Equation 2}
\end{align}
\hrulefill 
\end{figure*}
\end{proof}
\begin{remark}
Invoking Lemma \ref{Lemma 1} and the definition of inner product on $H$, the inner product in the definition of operator $\bar{L}^{i}$ could be evaluated as follows, 
\begin{align*}
&\bar{L}^{i}(f)(y) =\langle f(\cdot), K^{i}(\cdot,y) \rangle_{H} =\langle L^{1}(f), L^{1}(K^{i}(\cdot,y)) \rangle_{H^{1}} +  \\
&\langle L^{2}(f), L^{2}(K^{i}(\cdot,y)) \rangle_{H^{2}} \hspace{-3pt} =  \hspace{-3pt}  \langle \sum^{n}_{k=1}\beta_{k,f}K^{1}(\cdot, y_{k,f}),  \hspace{-2pt} \sum^{m}_{l=1}\beta_{l,K^{i}(\cdot,y)}\\
&K^{1}(\cdot, y_{l,K^{i}(\cdot,y)})\rangle_{H^{1}} + \langle \sum^{n}_{k=1}\beta_{k,f}K^{2}(\cdot, y_{k,f}),
\sum^{m}_{l=1}\beta_{l,K^{i}(\cdot,y)}\\
&K^2(\cdot, y_{l,K^{i}(\cdot,y)})\rangle_{H^{2}} =  \sum^{n}_{k=1}\sum^{m}_{l=1}\beta_{k,f}\beta_{l,K^{i}(\cdot,y)}K^{1}(y_{l,K^{i}(\cdot,y)}, \\
&y_{k,f}) + \sum^{n}_{k=1}\sum^{m}_{l=1}\beta_{k,f}\beta_{l,K^{i}(\cdot,y)}K^{2}(y_{l,K^{i}(\cdot,y)}, y_{k,f}).
\end{align*}
Unlike equation set (\ref{Equation 1}), the above evaluation does not result in a ``function" form for $\bar{L}^{i}(f)$. This is because in the above evaluation, the functions get mapped to their versions in the individual KSs and the inner product is evaluated there. To obtain a function form as in equation set (\ref{Equation 1}), it is necessary to evaluate the inner product in the fusion space.
\end{remark}
\begin{theorem}\label{Theorem 2}
Let $L: H \to H$ be a symmetric, positive, bounded operator. There exists a unique square root of operator $L$, $\sqrt{L}: H \to H$, i.e., $\sqrt{L} (\sqrt{L} (f)) = L(f) \forall f \in H$. $\sqrt{L}(\cdot)$ is linear, bounded, symmetric. 
\end{theorem}
\begin{proof}
Since $L$ is symmetric, from the spectral theorem, it follows that (i) there exists an orthonormal basis of $H$, $\{\varphi_{j}\}_{j \in \mathcal{I}}$, which are eigenvectors of $L$ ; (ii) the eigenvalues of $L$, $\{\lambda_{j}\}^{\mathcal{I}_{\lambda}}_{j=1}$, are real. Since $L$ is positive,  $\lambda_{j} \geq 0, 1 \leq j \leq \mathcal{I}_{\lambda}$. Let $f = \sum^{\mathcal{I}}_{j=1}\alpha_j \varphi_{j}, \{\alpha_j\}\subset \mathbb{R}$. $L(f) = L(\sum^{\mathcal{I}}_{j=1}\alpha_j \varphi_{j}) = \sum^{\mathcal{I}}_{j=1}\alpha_j\lambda_j\varphi_{j}$. From the construction of $H$, it follows that  $\mathcal{I} \leq \mathcal{I}^1 + \mathcal{I}^2$. Since some of the eigenvalues of $L$ could repeated, $\mathcal{I}_{\lambda} \leq \mathcal{I}$. The operator $\sqrt{L}$ is defined as $\sqrt{L}(\varphi_{j}) = \sqrt{\lambda_j}\varphi_{j}$, and, imposing linearity $\sqrt{L}(f) = \sum^{\mathcal{I}}_{j=1}\alpha_j\sqrt{\lambda_j} \varphi_{j}$.
\begin{align*}
&\sqrt{L}(\sqrt{L}(f)) = \sqrt{L}( \sum^{\mathcal{I}}_{j=1}\alpha_j\sqrt{\lambda_j} \varphi_{j}) =  \sum^{\mathcal{I}}_{j=1}\alpha_j\lambda_j\varphi_{j} = L(f) \\
&\langle \sqrt{L}(f), f \rangle_{H} \hspace{-3pt} = \hspace{-3pt} \langle \sum^{\mathcal{I}}_{j=1}\alpha_j\sqrt{\lambda_j} \varphi_{j},\hspace{-1pt}\sum^{\mathcal{I}}_{l=1}\alpha_l \varphi_{l}\rangle_{H} \hspace{-3pt} =  \hspace{-3pt} \sum^{\mathcal{I}}_{j=1} \alpha^2_j \sqrt{\lambda_j} \geq 0 \\
&\langle \sqrt{L}(f), g \rangle_{H} = \langle \sum^{\mathcal{I}}_{j=1}\alpha_j\sqrt{\lambda_j} \varphi_{j},\hspace{-1pt}\sum^{\mathcal{I}}_{l=1}\beta_l \varphi_{l}\rangle_{H}  \hspace{-1pt} = \hspace{-1pt}  \sum^{\mathcal{I}}_{j=1} \alpha_j\beta_j \sqrt{\lambda_j} \\
&=\langle \sum^{\mathcal{I}}_{j=1}\alpha_j\varphi_{j},\hspace{-1pt}\sum^{\mathcal{I}}_{l=1}\beta_l\sqrt{\lambda_l} \varphi_{l}\rangle_{H} = \langle f,  \sqrt{L}(g) \rangle_{H}. 
\end{align*}
Suppose $\tilde{L}$ is a positive  semi-definite operator which is another square root for $L$. To prove uniqueness, it suffices to prove that $\tilde{L}(\phi_{j}) = \sqrt{L}(\phi_j) \forall \;j$, where $\{\phi_j\}_{j \in \mathcal{I}}$ is a basis for $H$.  By the spectral theorem, $\tilde{L}$ posses a set of orthonormal eigenvectors,$\{\phi_j\}_{j \in \mathcal{I}}$, which form a basis for $H$. $L(\phi_{j}) = \tilde{L}\big(\tilde{L}(\phi_j)\big) = \tilde{L}(\bar{\lambda}_j\phi_j) = \bar{\lambda}^2_j\phi_j$. Thus, $\{\phi_j\}_{j \in \mathcal{I}}$ are eigenvectors for $L$. By definition of $\sqrt{L}$, $\{\phi_j\}_{j \in \mathcal{I}}$ are eigenvectors for $\sqrt{L}$ as well. As the eigenvectors $\{\phi_j\}_{j \in \mathcal{I}}$ could be ordered differently than  $\{\varphi_j\}_{j \in \mathcal{I}}$, we denote the corresponding eigenvalues by $\hat{\lambda}_j$. Thus, $\sqrt{L}\big(\sqrt{L}(\phi_j)\big) = \hat{\lambda}^2_j \phi_j = L(\phi_j) = \bar{\lambda}^2_j\phi_j$. Since the operators are positive semidefinite, $\hat{\lambda}_j = \bar{\lambda}_j$. Hence, $\tilde{L}(\phi_j) = \sqrt{L}(\phi_j), \forall j$, i.e. $\tilde{L} =\sqrt{L}$.
\end{proof}
The above theorem, Theorem \ref{Theorem 2}, is well known in linear algebra and we mention the proof as the construction of the square root operator is essential for the proof of the theorem below. 
\begin{theorem}\label{Theorem 3}
The linear space  $\bar{H}^{i} = \{g : g = \sqrt{\bar{L}^i}(f), f\in H \}$ is a RKHS with kernel $K^i$.  $\sqrt{\bar{L}^i}(\cdot)$ establishes an isometric isomorphism between $\mathcal{N}\big(\sqrt{\bar{L}^i}\big)^{\perp}$ and $\bar{H}^{i}$, and the norm, $||f||_{\bar{H}^{i}} = ||g||_{H}$,  where $f = \sqrt{\bar{L}^i} g,  g\in \mathcal{N}\big(\sqrt{\bar{L}^i}\big)^{\perp}$. Thus, the individual knowledge spaces can be retrieved from the fusion space. 
\end{theorem}
\begin{proof}
Since $\mathcal{N}\big(\sqrt{\bar{L}^i}\big)$ is closed subspace of $H$, $H = \mathcal{M}^{i} \oplus \mathcal{N}\big(\sqrt{\bar{L}^i}\big)$, where $\mathcal{M}^{i} = \mathcal{N}\big(\sqrt{\bar{L}^i}\big)^{\perp}$. $\sqrt{\bar{L}^i}()$ maps one-one from $\mathcal{M}^{i}$ to $\bar{H}^{i}$. Hence, $\mathcal{M}^{i} \subset \bar{H}^{i}$. Let $f \in \mathcal{N}\big(\sqrt{\bar{L}^i}\big)$ and $g \in H$. Then, $\langle f,  \sqrt{\bar{L}^i}(g) \rangle_{H} = \langle \sqrt{\bar{L}^i}(f),  g \rangle_{H} = 0$, i.e., $\sqrt{\bar{L}^i}(g) \perp f, \; \forall f \in \mathcal{N}\big(\sqrt{\bar{L}^i}\big)$, $\forall g \in H$. Hence, $\bar{H}^{i} \subset \mathcal{N}\big(\sqrt{\bar{L}^i}\big)^{\perp} = \mathcal{M}^{i}$. Therefore, $M^{i} = \bar{H}^{i}$. The inner product on $\bar{H}^{i}$ is defined as,
\begin{align*}
\langle f, g \rangle_{\bar{H}^{i}} =  \langle \sqrt{\bar{L}^i}(\bar{f}), \sqrt{\bar{L}^i}(\bar{g}) \rangle_{\bar{H}^{i}} = \langle \Pi_{\mathcal{M}^i} (\bar{f}), \Pi_{\mathcal{M}^i} (\bar{g}) \rangle_{H},
\end{align*}
where $f = \sqrt{\bar{L}^i}(\bar{f}), \bar{f} \in H $ and $g = \sqrt{\bar{L}^i}(\bar{g}), \bar{g} \in H$. For $f(\cdot) = K(\cdot,y)$, $\bar{L}^{i}(f)(\cdot) = K^{i}(\cdot,y)$. Hence, $K^{i}(\cdot, y) \in  \mathcal{R}(\bar{L}^i) \subseteq \bar{H}^{i},\; \forall y \in \mathcal{X}$. Let $\{\psi^{i}_j\}$ be the eigenvectors of $\bar{L}^{i}$ with corresponding eigenvalues, $\{\lambda^{i}_{j}\}$. Then, from the construction of $\sqrt{\bar{L}^i}$ in Theorem  \ref{Theorem 2}, $\{\psi^{i}_j\}$ are the eigenvectors of $\sqrt{\bar{L}^{i}}$ with corresponding eigenvalues, $\{\sqrt{\lambda^{i}_{j}}\}$. For any eigenvector, $\psi^{i}_j$, with $\lambda^i_j \neq 0$,
\begin{align*}
&\langle \psi^{i}_{j}(\cdot), K^{i}(\cdot,y) \rangle_{\bar{H}^{i}} = \\
&\langle \sqrt{\bar{L}^i}\Big(\frac{1}{\sqrt{\lambda^{i}_{j}}}\psi^{i}_{j}(\cdot) \Big),\sqrt{\bar{L}^i} \big( \sqrt{\bar{L}^i}(K(\cdot,y)) \big) \rangle_{\bar{H}^{i}} \\
&\overset{(a)}{=}\langle \frac{1}{\sqrt{\lambda^{i}_{j}}}\psi^{i}_{j}(\cdot), \Pi_{\mathcal{M}^i}(\sqrt{\bar{L}^i}(K(\cdot,y))) \rangle_{H} ,\overset{(b)}{=} \langle  \frac{1}{\sqrt{\lambda^{i}_{j}}}\psi^{i}_{j}(\cdot),\\
&\sqrt{\bar{L}^i}(K(\cdot,y))   \rangle_{H} \overset{(c)}{=} \langle \sqrt{\bar{L}^i}( \frac{1}{\sqrt{\lambda^{i}_{j}}}\psi^{i}_{j}(\cdot)), K(\cdot,y)\rangle_{H} \\
&=  \langle \psi^{i}_j(\cdot), K(\cdot,y)\rangle_{H} = \psi^{i}_j(y) ,
\end{align*}
Thus, the reproducing property is satisfied by $\psi^{i}_j(\cdot)$. Since every function in $\mathcal{M}^{i} = \bar{H}^{i}$ can expressed as unique linear combination of $\{\psi^{i}_{j}\}$, and by the linearity of inner product it follows that $\langle f(\cdot), K^{i}(\cdot,y) \rangle_{\bar{H}^{i}} = f(y), \; \forall f \in \bar{H^i}$. The reasoning for the equalities are as follows, $(a)$ by the definition of the inner product and $\psi^{i}_{j} \in \mathcal{M}^{i}$ as $\lambda^{i}_{j} \neq 0$, $(b) \;$ $\sqrt{\bar{L}^i}(K^{i}(\cdot,y)) \in \bar{H}^{i} =\mathcal{M}^{i}$, and $(c)$ symmetry of $\sqrt{\bar{L}^i}$. From Definition \ref{definition 2}, it follows that $\bar{H}^{i}$ is a RKHS with kernel $K^{i}(\cdot,\cdot)$. 
\end{proof}
\begin{corollary}
The downloading operator from the fusion space $H$ to  agent $i$'s knowledge space, $H^i$, is $\sqrt{\bar{L}^{i}} \circ \Pi_{\mathcal{M}^i}$. The downloading operator is linear and bounded.
\end{corollary}
\section{Regression and Fusion Problem}\label{section 3}
The knowledge spaces constructed in the previous section can used to formulate many inference problems including regression, classification, etc. In this section, we consider the least squares regression problem and the fusion problem as an optimization problem. 
\subsection{Regression at the Agents}\label{subsection 3.1}
Given (input-output) information pairs, $\{(x^i_j, y^i_j)\}^{m}_{j=1}$, to agent $i$, the objective of the agent is to estimate a mapping from input to output. The estimation problem formulated as least squares regression problem is 
\begin{align*}
\underset{f \in H^{i}} \min \sum^{m}_{j=1}(y^i_j-f(x^i_j))^2 + \varrho^{i}|| f ||^{2}_{H^i}.
\end{align*}
Let, $\mathbf{K^{i}} = (K^{i}(x^{i}_{j}, x^{i}_{k}))_{jk}= (\langle K^{i}(\cdot, x^i_k), K^i(\cdot, x^i_{j})  \rangle_{H^i})$, be the Gram matrix corresponding to agent $i$ as defined in the beginning of Section \ref{section 2}. It is well known from the representer theorem (refer \cite{hofmann2008kernel}), that the solution for the above problem is given by $f^i(\cdot) = \sum^{m}_{l=1}\alpha^{i}_{l}K^{i}(\cdot,x^{i}_l)$ where $\boldsymbol{\alpha^{i}}=(\mathbf{K^i}^T\mathbf{K^i} + \varrho^i \mathbf{K^i})^{-1}\mathbf{K^i}^{T}\mathbf{y^i}$ and $\mathbf{y^i} = (y^i_{1},\ldots, y^i_{m})$.  
\subsection{Function Fusion Problem}\label{subsection 3.2}
The functions estimated by the agents are transmitted to the fusion center where the following fusion problem is considered. As presented in \cite{raghavan2023distributed}, we consider $\{\mathfrak{b} = \{b_{k}\}_{k\geq 1} \subset H \}$ which span $H$ to define a dissimilarity measure between $f^{1}$ and $f^2$ as:
\begin{align*}
d_{\mathfrak{b}}(f, g)  =\sum_{k}\langle f - g, b_k \rangle^2_{H}.
\end{align*}
The fusion problem is to find a linear combination of $f^{1}$ and $f^2$,  $f^*$,  such that the dissimilarity between $f^1,  f^*$ and $f^2, f^*$ is minimized.  The fusion problem as an optimization problem is 
\begin{align*}
 \; \underset{a,b\in \mathbb{R}}\min d_{\mathfrak{b}}(af^1+bf^{2},  f^{1}) + d_{\mathfrak{b}}(af^1+bf^{2},  f^{2}) + \\
\varrho|| af^1+bf^{2}||^{2}_{H}.
\end{align*}
The fused function is considered as the function estimated by the system. It is downloaded by the agents to compare (in the sense of norm) against their own estimates.
\section{Example}\label{section 4}
In this section, we demonstrate the application of the theory developed in the previous sections. We consider the estimation of a real valued cubic polynomial with real valued inputs. The coefficients of the polynomial where chosen at random. Agent 1 was provided with $20$ samples of input-output data, where the input was restricted to the set $[-5,5]$. The inputs were uniformly spaced on the interval $[-5,5]$  and the corresponding outputs were obtained by providing the inputs to the true function. The features considered by Agent 1 were, $\varphi^{1}_{1}(x)=1, \varphi^{1}_{2}(x)= x$, and $\varphi^{1}_{2}(x)= x^2$, which implies that its kernel was $K^{1}(x,y) = 1+xy+x^2y^2$. Agent 2 was also provided with $20$ samples of input-output data where the input data was uniformly spaced and restricted to $[-10,-5] \cup [5,10]$, while the features considered by it where $\varphi^{2}_{1}(x)=x^2$ and $\varphi^{2}_{2}(x)= x^3$. Hence, $K^{2}(x,y) = x^2y^2 + x^3y^3$. 

With this set up, the regression problems (subsection \ref{subsection 3.1}) were solved by the agents. These functions were uploaded to the fusion center. Function uploaded by Agent 1 and Agent 2 are plotted in  Figure \ref{Figure 2} and Figure \ref{Figure 3} respectively. At the fusion center, the set $\mathfrak{b} = \{ K(\cdot,\bar{x}_{j})\}^{40}_{j =1}$ was considered, where $\{\bar{x}_{j}\}^{40}_{j=1}$ were randomly sampled from $[-10,10]$. The function fusion problem (subsection \ref{subsection 3.2}) was solved. The fused function is plotted in Figure \ref{Figure 4}.
 
To download the fused function onto the individual KSs, we demonstrate the procedure outlined in subsection \ref{subsection 2.3}. Suppose we choose the set of basis vectors for the space $H$ as $\varphi_1(x)=1, \varphi_2(x) = x, \varphi_3(x) = \sqrt{2}x^2 ,\varphi_4(x) =x^3$, then the kernel generated is $K(x,y) = 1 + xy + 2x^2y^2 +x^3y^3 = K^{1}(x,y) + K^{2}(x,y)$. With these basis vectors, the coefficients for $K^{1}(\cdot,y)$ are $[1,y, \frac{y^2}{\sqrt{2}}, 0]$, and for $K^{2}(\cdot,y)$ are $[0,0, \frac{y^2}{\sqrt{2}}, y^3]$. Thus, the matrix representation for $\bar{L}^{1}$ is obtained as follows:
\begin{align*}
&\bar{L}^{1}(\varphi_1)(y)  \hspace{-2pt} = \hspace{-2pt} \langle \varphi_1(\cdot), K^{1}(\cdot,y) \rangle_{H} \hspace{-2pt} =\hspace{-3pt} \langle [1,0,0,0], [1,y, \frac{y^2}{\sqrt{2}}, 0] \rangle_{\mathbb{R}^4}\\
&=1. \;\hspace{-2pt}\bar{L}^{1}(\varphi_2)(y) \hspace{-2pt}  = \hspace{-2pt} \langle [0,1,0,0], [1,y, \frac{y^2}{\sqrt{2}}, 0] \rangle_{\mathbb{R}^4} \hspace{-2pt} = \hspace{-2pt}y.\; \hspace{-2pt}\bar{L}^{1}(\varphi_3)(y)\\
&= \langle [0,1,0,0], [1,y, \frac{y^2}{\sqrt{2}}, 0] \rangle_{\mathbb{R}^4} =\frac{y^2}{\sqrt{2}}.\; \bar{L}^{1}(\varphi_4)(y) = 0.\\
&L^{1}_{M}= 
\begin{bmatrix}
1 & 0 & 0 & 0\\
0 & 1 & 0 & 0 \\
0 & 0 & \frac{1}{2} & 0 \\
0 & 0 & 0 & 0 
\end{bmatrix}, \; 
\sqrt{L^{1}_{M}}= 
\begin{bmatrix}
1 & 0 & 0 & 0\\
0 & 1 & 0 & 0 \\
0 & 0 & \frac{1}{\sqrt{2}} & 0 \\
0 & 0 & 0 & 0 
\end{bmatrix},\\
&L^{2}_{M}= 
\begin{bmatrix}
0 & 0 & 0 & 0\\
0 & 0 & 0 & 0 \\
0 & 0 & \frac{1}{2} & 0 \\
0 & 0 & 0 & 1
\end{bmatrix}, \; 
\sqrt{L^{2}_{M}}= 
\begin{bmatrix}
0 & 0 & 0 & 0\\
0 & 0 & 0 & 0 \\
0 & 0 & \frac{1}{\sqrt{2}} & 0 \\
0 & 0 & 0 & 1 
\end{bmatrix}.
\end{align*}
\begin{figure}
\begin{center}
\includegraphics[width=\columnwidth]{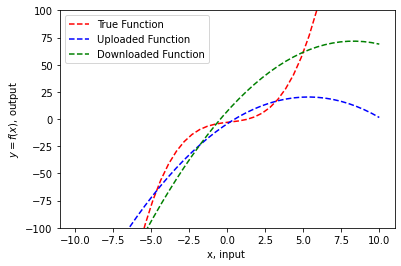}
\caption{True function, function uploaded by Agent 1, function downloaded by Agent 1.} 
\label{Figure 2}
\end{center}
\vspace{-0.7cm}
\end{figure}
\begin{figure}
\begin{center}
\includegraphics[width=\columnwidth]{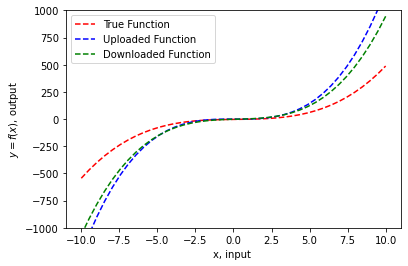}
\caption{True function, function uploaded by Agent 2, function downloaded by Agent 2.} 
\label{Figure 3}
\end{center}
\vspace{-0.7cm}
\end{figure}
\begin{figure}
\begin{center}
\includegraphics[width=\columnwidth]{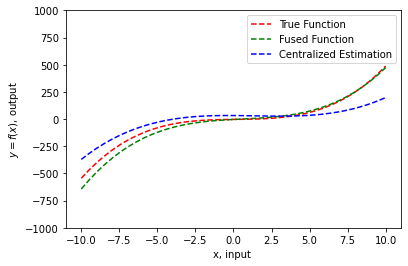}
\caption{True function, fused function, function obtained through centralized estimation at fusion center.} 
\label{Figure 4}
\end{center}
\vspace{-0.8cm}
\end{figure} 
The coefficients of the fused function with respect to basis chosen for $H$ were obtained. The coefficients for downloaded functions were obtained through the operation $\boldsymbol{\alpha}^{i} = \sqrt{L^{i}_{M}}\boldsymbol{\alpha}$ where $\boldsymbol{\alpha}$ is the vector of coefficients corresponding to the fused function. The downloaded functions correspond to $\sum^{4}_{j=1}\alpha^{i}_{j}\varphi_j$, where $\boldsymbol{\alpha}^{i} = [\alpha^{i}_{1},\alpha^{i}_{2},\alpha^{i}_{3},\alpha^{i}_{4} ]$. The function downloaded by Agent 1 and Agent 2  are plotted in Figures \ref{Figure 2} and \ref{Figure 3} respectively. 

We observe that both the uploaded and downloaded function for Agent 1 do not estimate the true function well as it is missing the $\varphi_{4}(x)=x^{3}$ feature. Agent $2$ is able to ``better" estimate the function however the impact of missing  data points ( between $[-5,5]$ ) and missing features is visible. However, we note that the agents are able to exchange knowledge as the downloaded functions are ``significantly" better than the uploaded functions. To compare the performance of the fusion procedure, we compare it against the centralized estimation method. In this method, the data collected by both agents is sent to the fusion center, where a regression problem (refer subsection \ref{subsection 3.1}) is solved using the collective data.  The function estimated using this method is plotted in Figure \ref{Figure 4}. 

We observe that the fused function is the best estimate among all estimates considered so far. The centralized scheme is where the data collected by both the agents is simultaneously processed in the fusion space. In this particular example, it is possible that the joint processing of the data using both the kernels led to ``poor" computations as compared to local estimation using individual kernels and averaging of the local estimates in the fusion space. Further theoretical analysis is needed to compare the fusion method and the centralized  estimation method.  
\section{Conclusion and Future Work}\label{section 5}
We considered the problem of function estimation by two agents given local data and different set of features. We presented the construction of suitable spaces for the estimation problem and fusion problems to be studied. We derived operators to transform functions across spaces coherently.  A distributed estimation scheme without exchange of data was presented to solve the problem considered. As future work, we are interested in: (i) developing a sequential collaborative learning scheme involving the agents and the fusion center; (ii) studying the consistency properties of such a scheme in the local KS and the fusion space; (iii) quantifying the transfer of knowledge from one agent to another. 
\bibliographystyle{IEEEtran}
\bibliography{biblio}
\end{document}